\documentclass[letterpaper, 10 pt, conference]{ieeeconf}  

\IEEEoverridecommandlockouts                              

\let\labelindent\relax

\usepackage{grffile}
\usepackage{amsmath,amssymb,amsbsy}
\usepackage{hyperref}
\usepackage{amsthm} 
\usepackage{color}
\usepackage{comment}
\usepackage{lipsum}
\usepackage{graphicx}%
\usepackage[normalem]{ulem}
\usepackage{enumerate}
\usepackage{epstopdf}
\usepackage{subcaption}

\newtheorem{theorem}{Theorem}[section]
\newtheorem{lemma}[theorem]{Lemma}
	
\newtheorem{proposition}[theorem]{Proposition}
	
\newtheorem{remark}[theorem]{Remark}	
\newtheorem{problem}[theorem]{Problem}

\allowdisplaybreaks

\title{\LARGE \bf Optimal Transport over Deterministic  Discrete-time Nonlinear Systems using Stochastic Feedback Laws
}

\author{Karthik Elamvazhuthi$^{1}$, Piyush Grover$^{2}$, and Spring Berman$^{1}$
\thanks{The first and third authors were supported in part by ONR Young Investigator Award N00014-16-1-2605. The first author was supported in part by MERL. The second author was solely supported by MERL.}
\thanks{$^{1}$Karthik Elamvazhuthi and Spring Berman are with the School for Engineering of
	Matter, Transport and Energy, Arizona State University, Tempe, AZ, 85287
	USA {\tt\small \{karthikevaz, Spring.Berman\}@asu.edu}.}%
\thanks{$^{2}$Piyush Grover is with Mitsubishi Electric Research Laboratories (MERL),
Cambridge, MA 02139
        {\tt\small grover@merl.com}.}%
}

\begin{document}

\maketitle
\thispagestyle{empty}
\pagestyle{empty}

\begin{abstract}
This paper considers the relaxed version of the {\it transport problem} for general nonlinear control systems, where the objective is to design time-varying feedback laws that transport a given initial probability measure to a target probability measure under the action of the closed-loop system. To make the problem analytically tractable, we consider control laws that are {\it stochastic}, i.e., the control laws are maps from the state space of the control system to the space of probability measures on the set of admissible control inputs. Under some controllability assumptions on the control system as defined on the state space, we show that the transport problem, considered as a controllability problem for the lifted control system on the space of probability measures, is well-posed for a large class of initial and target measures. We use this to prove the well-posedness of a fixed-endpoint optimal control problem defined on the space of probability measures, where along with the terminal constraints, the goal is to optimize an objective functional along the trajectory of the control system. This optimization problem can be posed as an infinite-dimensional linear programming problem. This formulation facilitates numerical solutions of the transport problem for low-dimensional control systems, as we show in two numerical examples.
\end{abstract}

\section{INTRODUCTION}

In this paper, we consider a variation of the {\it optimal transport problem} \cite{villani2008optimal}. The 
objective of this problem
is to construct a map such that a given probability measure is {\it pushed forward} to a target probability measure in some optimal manner. Initially motivated by resource allocation problems in economics, this problem has potential applications in many engineering problems involving the control of large-scale distributed systems \cite{djehiche2016mean}, in which these measures could represent the distribution of an ensemble of agents such as a swarm of robots \cite{elamvazhuthi2015optimal} or the distribution of nodes in an electric power grid \cite{bagagiolo2014mean} or a wireless network \cite{tembine2014energy}. For example, we have employed this modeling approach in the design and experimental validation of stochastic coverage and task allocation strategies for swarms of robots \cite{deshmukh2018mean}.

In the original formulation of optimal transport, 
the dynamics of the agents are simplistic from a control-theoretic point of view. There have been some recent efforts to extend classical optimal transport theory to the case where the cost functions and transport maps are subject to dynamical constraints arising from control systems.     
Toward this end, \cite{hindawi2011mass} considers the optimal transport problem for linear time-invariant systems with linear quadratic cost functions. For a smaller class of cost functions, the case of linear time-varying systems is addressed in \cite{chen2017optimal}. 
There have also been efforts to extend
the theory to nonlinear driftless control-affine systems in the framework of {\it sub-Riemannian optimal transport} \cite{agrachev2009optimal,figalli2010mass,khesin2009nonholonomic}. See also \cite{elamvazhuthi2016optimal}, in which we develop connections between computational optimal transport over continuous-time nonlinear control systems and optimal transport on finite state spaces. Closely related to such optimal transport problems is the theory of mean-field games and mean-field type controls \cite{bagagiolo2014mean,djehiche2016mean,tembine2014energy}.

The original optimal transport problem, i.e., the {\it Monge problem}, searches for a deterministic map that maps a given measure to a target measure. In view of the analytical difficulties involved in this original formulation 
of Monge, Kantorovich introduced a relaxed version of the problem in 1942, in which the map is allowed to be stochastic. 
This form of relaxation, which is used to convexify nonlinear control problems, has a rich history in control theory in the context of {\it Young measures} or {\it relaxed control} \cite{florescu2012young}. Such a measure-based convexification of optimization problems has been used for numerical synthesis of control laws \cite{hernandez2012discrete,lasserre2008nonlinear,raghunathan2014optimal}. 

In this paper, we use a similar relaxation procedure to consider the optimal transport problem for discrete-time nonlinear control systems with a compact set of admissible controls. Before considering the issue of optimality, we consider the problem of controllability. First, we prove that controllability of the original control system implies controllability of the control system induced on the space of probability measures. Next, we show that we can frame the control-constrained optimal transport problem of controllable nonlinear systems as a linear programming problem, as in the Kantorovich formulation of the optimal transport problem. Unlike our previous work \cite{elamvazhuthi2016optimal}, which focused on computational aspects of optimal transport problems for nonlinear systems with a particular control-affine structure, in this paper we solve the optimal transport problem for general nonlinear control systems in discrete time. 


\section{NOTATION AND TERMINOLOGY}

Let $X$ be a separable finite-dimensional manifold (for example, the Euclidean space $\mathbb{R}^M$) that is a metric space. The set of admissible control inputs will be denoted by $U$. We will assume that the set $U$ is a compact subset of a metric space. Note that $X \times U$, equipped with the product topology,  is a metrizable and separable space under these assumptions. We will denote by $\mathcal{B}(X)$, $\mathcal{B}(U)$, and $\mathcal{B}(X \times U)$ the collection of Borel measurable sets of $X$, $U$, and $X \times U$, respectively. The space of Borel probability measures on the sets $X$ and $U$ will be denoted by $\mathcal{P}(X)$ and $\mathcal{P}(U)$, respectively. For a metric space $Y$, let $C_b(Y)$ be the set of bounded continuous functions on $Y$. We will say that a sequence of measures $(\mu_n)_{n=1}^{\infty} \in \mathcal{P}(Y)$ converges narrowly to a limit measure $\mu \in \mathcal{P}(Y)$ if the sequence $\int_{Y} f(\mathbf{y}) d\mu_n(\mathbf{y})$ converges to $\int_{Y} f(\mathbf{y}) d\mu(\mathbf{y})$ for every $f \in C_b(Y)$. The topology on $\mathcal{P}(Y)$ corresponding to this convergence will be referred to as the narrow topology. For a set $M \subset X$ and $p \in \mathbb{Z}_+$, we will define the set $ D^p_M = \big \lbrace \sum_{i=1}^p c_i\delta_{\mathbf{y}_i}; ~\mathbf{y}_i \in M, ~ c_i \in [0,1] ~ \text{for} ~ i  \in \lbrace 1,...,p \rbrace, ~ \sum_{i = 1}^pc_i=1 \big \rbrace$, where $\delta_{\mathbf{x}}$ is the Dirac measure concentrated at the point $\mathbf{x} \in X$. We will also define the set $D_M = \cup_{p \in \mathbb{Z}_+} D_M^p$. The support of a measure $\mu \in \mathcal{P}(X)$ will be denoted by ${\rm supp} ~ \mu = \lbrace \mathbf{x} \in X; ~\mathbf{x} \in N_\mathbf{x} ~  {\rm implies~that} ~ \mu(N_\mathbf{x}) >0, ~{\rm where~} N_\mathbf{x} ~ {\rm  is~a~ neighborhood ~ of~} \mathbf{x}  \rbrace$. We define $\mathcal{Y}(X,U)$ as the set of stochastic feedback laws, i.e., maps of the form $K: X \times \mathcal{B}(U) \rightarrow \mathbb{R}$, where $K(\cdot,A)$ is Borel measurable for each $A \in \mathcal{B}(U)$ and $K(\mathbf{x},\cdot) \in \mathcal{P}(U)$ for each $\mathbf{x}\in X$. For a continuous map $F : Y \rightarrow X$, the pushforward map $F_{\#}:\mathcal{P}(Y) \rightarrow \mathcal{P}(X)$ is defined by 
\begin{align*}
(F_{\#}\mu) (A) = \mu(F^{-1}(A)) = \int_Y \mathbf{1}_A(F(\mathbf{y}))d\mu(\mathbf{y})
\end{align*}
for each $A \in \mathcal{B}(X)$, where $\mathbf{1}_B$ denotes the indicator function of the set $B \in \mathcal{B}(X)$ and $\mu \in \mathcal{P}(Y)$. 

\section{PROBLEM FORMULATION}

Now we are ready to state the problems addressed in this paper. Consider the nonlinear discrete-time control system
\begin{equation} \label{eq:origsys}
\mathbf{x}_{n+1} = T(\mathbf{x}_n,\mathbf{u}_n), ~ n = 0,1,...~; ~~~~\mathbf{x}_0 \in X, 
\end{equation}
where $\mathbf{x}_n \in X$ for each $n \in \mathbb{Z}_{+}$, $(\mathbf{u}_i)_{i=0}^\infty$ is a sequence in a compact set $U$, and $T:X \times U \rightarrow X$ is a continuous map with respect to the topologies $\mathcal{T}(X)$, $\mathcal{T}(U)$, and $\mathcal{T}(X) \times \mathcal{T}(U)$ defined on $X$, $U$, and $X \times U$, respectively.  Then this nonlinear control system induces a control system on the space of measures $\mathcal{P}(X)$, given by
\begin{equation}
\mu_{n+1} = T(\cdot,\mathbf{u}_n)_{\#}\mu_n,~ n = 0,1,...~; ~~ \mu_0 \in \mathcal{P}(X). \label{eq:indctrlsys}
\end{equation}

The first problem of interest is the following.

\begin{problem} \label{reachprob0}
	(\textbf{Controllability problem with deterministic control}). Let $N \in \mathbb{Z}_+$ be a specified final time. Given an initial measure $\mu_0 \in \mathcal{P}(X)$ and a target measure $\mu^f \in \mathcal{P}(X)$, does there exist a sequence of feedback laws $\mathbf{v}_n:X \rightarrow U$ such that the closed-loop system satisfies
	\begin{equation*}
	\mu_{n+1} = T^{cl,n}_{\#}\mu_n,  ~~ n=0,1,...,N-1; ~~~~
	 \mu_{N} = \mu^f,
	\end{equation*}
where $T^{cl,n}_{\#}:\mathcal{P}(X) \rightarrow \mathcal{P}(X)$ is the pushforward map corresponding to the closed-loop map $T^{cl,n} : X \rightarrow X$ defined by $T^{cl,n}(\mathbf{x}) = T(\mathbf{x},\mathbf{v}_n(\mathbf{x}))$ for all $\mathbf{x} \in X$?
\end{problem}

This problem is unsolvable in general. For instance, consider the case when $X=\mathbb{R}$, $U = [-1,1]$, $T(\mathbf{x},\mathbf{u}) = \mathbf{x}+\mathbf{u}$ for each $(\mathbf{x},\mathbf{u}) \in X \times U$, $\mu_0 = \delta_0$ is the Dirac measure concentrated at the point $0 \in \mathbb{R}$, and $\mu^f = \frac{1}{2}\delta_{-1} +\frac{1}{2}\delta_{+1}$ is the sum of Dirac measures concentrated at $-1$ and $1$, respectively. This example does not admit any solutions to the controllability problem because a deterministic map cannot take the measure concentrated at the point $0$ and distribute it onto measures concentrated at $-1$ and $+1$. However, there might be several important cases where the problem does admit a solution. For example, when $X = \mathbb{R}^M$, $U = \mathbb{R}^M$ (which is not compact, in contrast to the assumptions made in this paper),  $T(\mathbf{x},\mathbf{u}) = \mathbf{u}$ for all $(\mathbf{x},\mathbf{u}) \in X \times U$, and $N=1$, this problem is equivalent to the classical optimal transport problem \cite{villani2008optimal}, for which solutions are known to exist when the initial and final measures are absolutely continuous with respect to the Lebesgue measure and have a finite second moment. On the other hand, this problem is expected to be highly challenging for general nonlinear control systems without any further constraints on the control set $U$, which is only assumed to be compact, given a final time $N \geq 1 $. Hence, to make the problem analytically tractable, we consider the following relaxed problem. 

\begin{problem} \label{reachprob}
(\textbf{Controllability problem with stochastic control}) Given a final time $N \in \mathbb{Z}_+$, an initial measure $\mu_0 \in \mathcal{P}(X)$, and a target measure $\mu^f \in \mathcal{P}(X)$, determine whether there exists a sequence of stochastic feedback laws $K_n \in \mathcal{Y}(X,U)$ such that the closed-loop system satisfies
\begin{equation}
\label{eq:ctrbcond}
\mu_{n+1} = T^{cl,n}_{\#}\mu_n  , ~~ n =0,1,...,N-1;  ~~~\mu_{N} = \mu^f, 
\end{equation}
where the closed-loop pushforward map $T^{cl,n}_{\#}$ is given by 
\begin{align}
(T^{cl,n}_{\#}\mu)(A) &= \int_X \int_U \mathbf{1}_A(T(\mathbf{x},\mathbf{u}))K_n(\mathbf{x},d\mathbf{u})d\mu(\mathbf{x}).
\end{align}
\end{problem}
Problem \ref{reachprob} can be considered a relaxation of Problem \ref{reachprob0} in the sense that deterministic control laws $\mathbf{v} :X \rightarrow U$ are just special types of stochastic control laws identified through the mapping $\mathbf{v}(\mathbf{x})   \mapsto \delta_{\mathbf{v}(\mathbf{x})}$. 

After addressing Problem \ref{reachprob}, we will address the following optimization problem.
\begin{problem} \label{OCP}
(\textbf{Fixed-time, fixed-endpoint optimal control problem})
Suppose that
 $ c:X \times U \rightarrow \mathbb{R}$ is a continuous map.   Given a final time $N \in \mathbb{Z}_+$, an initial measure $\mu_0 \in \mathcal{P}(X)$, and a target measure $\mu^f \in \mathcal{P}(X)$,
 determine whether the following optimization problem admits a solution:
 \begin{eqnarray}
 \label{eq:op1main1}
 \min_{\substack{\mu_m \in \mathcal{P}(X) \\ K_m \in \mathcal{Y}(X,U)}}~~ \sum_{m=0}^{N -1}\int_{X} \int_{U} c(\mathbf{x},\mathbf{u})K_m(\mathbf{x},d\mathbf{u})d\mu_m(\mathbf{x}) 
 \end{eqnarray}
subject to the constraints
\begin{equation}
\mu_{n+1} = T^{cl,n}_{\#}\mu_n  , ~~ n =0,1,...,N-1; ~~~\mu_{N} = \mu^f. 
\label{eq:op1main2}
\end{equation}
\end{problem}

Note that the control problem solved in this paper can be considered an extension of the problem addressed in \cite{raghunathan2014optimal}, in which the target measure is a Dirac measure. On the other hand, we consider more general target measures, but only address a finite-horizon optimal control problem.

\section{CONTROLLABILITY ANALYSIS}

In this section, we will address Problem \ref{reachprob}. Toward this end, we present the following definitions, which will be needed to define sufficient conditions under which Problem \ref{reachprob} 
admits a solution. Let $R^\mathbf{x}_1 = \lbrace T(\mathbf{x},\mathbf{u});~\mathbf{u} \in U \rbrace$ be the set of reachable states from $\mathbf{x} \in X$ at the first time step. Then we inductively define the set $R^\mathbf{x}_m = \cup_{\mathbf{y} \in R^\mathbf{x}_{m-1}} \lbrace   T(\mathbf{y},\mathbf{u});~ \mathbf{u} \in U \rbrace$ for each $m \in \mathbb{Z}_+ - \lbrace 1 \rbrace $. 


Instead of proving that we can always find a sequence of stochastic feedback laws $K_n$ such that the system of equations \eqref{eq:ctrbcond} is satisfied, 
we will consider the alternative ``convexified problem'' in which we look for measures $\nu_n$ in the space $\mathcal{P}(X \times U)$ such that, for given initial and target measures $\mu_0, \mu^f \in \mathcal{P}(X)$, the following constraints are satisfied:
\begin{equation}\label{eq:Modcont}
\mu_{n+1} = T_{\#} \nu_{n}, ~~ n = 0,1,...,N-1; ~~~~\mu_{N} = \mu^f,
\end{equation}
with $\nu_{n}(A \times U) = \mu_{n}(A)$ for all $A \in \mathcal{B}(X)$. 
We will first solve Problem \ref{reachprob} for the special case of Dirac measures, and then extend the result to general measures using a density-based argument that is standard in measure-theoretic probability. 


Now we are ready to present several results that address Problem \ref{reachprob}.

\begin{proposition}
	\label{DiracCtrb}
Let $\mu_0 = \delta_{\mathbf{x}_0}$ for some $\mathbf{x}_0 \in X$. Let $ \mu^f \in D^p_{M}$ for a compact subset $M$ of $X$, for some $p \in \mathbb{Z}_+$, such that ${\rm supp} ~ \mu^f \subseteq  R_N^{\mathbf{x}_0}$. Then there exists a sequence of measures $(\nu_m)_{m=0}^{N-1} \in \mathcal{P} (X \times U)$ such that 
\begin{equation}
\mu_{n+1} = T_{\#} \nu_{n}, ~~ n = 0,1,...,N-1,
\end{equation}
with $\nu_{n}(A \times U) = \mu_{n}(A)$ for all $A \in \mathcal{B}(X)$ and $\mu_{N} = \mu^f$.
\end{proposition}
\begin{proof}
Let $ \mu^f = \sum_{i=1}^p c^i\delta_{\mathbf{y}^i}$, where $\sum_{i=1}^pc^i=1$, for some $\mathbf{y}^i \in X$. By assumption, ${\rm supp} ~ \mu^f \subseteq R^{\mathbf{x}_0}_N$. Hence, for each $i \in \lbrace 1,...,p \rbrace$, there exists a sequence of inputs $(\mathbf{u}^i)_{n=0}^N$ such that the nonlinear discrete-time control system
\begin{equation}
\mathbf{x}^i_{n+1} = T(\mathbf{x}_n^i,\mathbf{u}_n^i), ~~ n = 0,1,...,N-1; ~~~\mathbf{x}^i_0 =\mathbf{x}^0
\end{equation}
satisfies $\mathbf{x}_N = \mathbf{y}^i$ for all $i \in \lbrace 1,...,p \rbrace $. We define $\nu^i_n = \delta_{(\mathbf{x}_{n-1}^i,\mathbf{u}_n^i)} \in \mathcal{P}(X \times U)$. Note that $(T_{\#} \nu^i_n)(A) = \delta_{\mathbf{x}^i_n}(A) $ for all $A \in \mathcal{B}(X)$ and all $i \in \lbrace 1,...,p \rbrace$. Then the result follows from the linearity of the operator $T_{\#}:\mathcal{P}(X \times U) \rightarrow \mathcal{P}(X)$ by setting $\nu_{n} =\sum_{i=1}^p c^i \nu^i_{n}$ for all $ n\in \lbrace 0,1,...,N-1 \rbrace$. In particular, for this choice of $\nu_n$, we have that $(T_{\#} \nu_{n}) = \sum_{i=1}^p c^i \mu^i_{n+1}$ for each $ n \in \lbrace 0,1,...,N-1 \rbrace $, and hence that $(T_{\#} \nu_{N-1}) = \sum_{i=1}^p c^i\delta_{\mathbf{y}^i} =  \mu^f$. 
\end{proof}

The next result follows immediately from Proposition \ref{DiracCtrb}.

\begin{lemma}
	\label{duadirac}
Let $\mu_0 \in  D^p_{A}$ and $ \mu^f \in D^q_{A}$ for a compact subset $A$ of $X$, for some $p,q \in \mathbb{Z}_+$, such that ${\rm supp} ~ \mu^f \subseteq  R_N^{\mathbf{x}}$ for each $\mathbf{x} \in {\rm supp} ~ \mu_0$. Then there exists a sequence of measures $(\nu_m)_{m=0}^{N-1} \in \mathcal{P} (X \times U)$ such that 
	\begin{equation}
	\mu_{n+1} = T_{\#} \nu_{n}, ~~ n = 0,1,...,N-1, 
	\end{equation}
	with $\nu_{n}(A \times U) = \mu_{n}(A)$ for all $A \in \mathcal{B}(X)$, and $\mu_{N} = \mu^f$.
\end{lemma}
\begin{proof}
Let $ \mu_0 = \sum_{i=1}^p c^i\delta_{\mathbf{y}^i}$, where $\sum_{i=1}^pc^i=1$, for some $\mathbf{y}^i \in X$. By assumption, ${\rm supp} ~ \mu^f  \subseteq \cap_{i=1}^p R^{\mathbf{y}^i}_N$. From Proposition \ref{DiracCtrb}, there exist measures $\nu_n^i \in \mathcal{P}(X \times U)$ such that if $\eta^i_0 = \mu_0$, then
\begin{equation}
\eta_{n+1}^i = T_{\#} \nu^i_{n}, ~~ n = 0,1,...,N-1, 
\end{equation}
with $\nu_{n}^i(A \times U) = \eta_{n}^i(A)$ for all $A \in \mathcal{B}(X)$, and $\eta^i_{N} = \mu^f$. The result follows by setting $\nu_{n} = \sum_{i =1}^{p}c^i \nu^i_{n}$ for all $ n\in \lbrace 0,1,...,N-1 \rbrace$.
\end{proof}
In order to prove the next proposition, we recall a well-known result, which follows from \cite{pedersen2012analysis}[Proposition 2.5.7], that  probability measures can be approximated using linear combinations of Dirac measures.



\begin{theorem}
	\label{dense_LC}
Let $Y$ be a locally compact Hausdorf space $Y$. Then the set of elements in $\mathcal{P}(Y)$ with support contained in a compact subset $M \subseteq Y$ is a convex and narrowly compact subset of $\mathcal{P}(Y)$. Additionally, the set $D_M$ is narrowly dense in the subset of $\mathcal{P}(Y)$ with supports contained in $M$.
\end{theorem}

\begin{proposition}
	\label{prodCtrb}
	Let $\mu_0, \mu^f \in \mathcal{P}(X)$ be Borel probability measures with compact supports, such that ${\rm supp} ~ \mu^f \subseteq  R_N^{\mathbf{x}}$ for each $ \mathbf{x} \in {\rm supp} ~ \mu_0$. Then there exists a sequence of measures $(\nu_m)_{m=0}^{N-1} \in \mathcal{P} (X \times U)$ such that 
	\begin{equation}
	\mu_{n+1} = T_{\#} \nu_{n}, ~~ n = 0,1,...,N -1,
	\end{equation}
	with $\nu_{n}(A \times U) = \mu_{n}(A)$ for all $A \in \mathcal{B}(X)$, and $\mu_{N} = \mu^f$.
\end{proposition}
\begin{proof}
Let $A  =  \cap_{ \mathbf{x} \in {\rm supp} ~\mu_0}R^\mathbf{x}_N \cup {\rm supp~ \mu_0}$. Clearly, the set $A$ is compact. From Theorem \ref{dense_LC}, we know that there exist sequences of measures $(\mu_0^i)_{i=1}^{\infty}, (\mu^{f,i})_{i=1}^{n\infty}  \in D_A$ such that $(\mu_0^i)_{i=1}^{\infty}$ and $ (\mu^{f,i})_{i=1}^{\infty}$ narrowly converge to $\mu_0$ and $\mu^f$, respectively. Then it follows from Lemma \ref{duadirac} that there exists a sequence of probability measures $(\nu^i_n)_{i=1}^{\infty}$ in $\mathcal{P}(X \times U)$ such that   
\begin{equation}
\mu^i_{n+1} = T_{\#} \nu^i_{n}, ~~ n = 0,1,...,N -1,
\end{equation}
with $\nu^i_{n}(A \times U) = \mu^i_{n}(A)$ for all $A \in \mathcal{B}(X)$ and $\mu^i_{N} = \mu^{f,i}$ for all $i \in \mathbb{Z}_+$. Since the map $T:X \times U \rightarrow X$ is continuous, the support of the measures $(\mu^i_{n+1},\nu^i_{n})$ is contained in a compact set for all $n \in \lbrace 0,1,...,N-1 \rbrace$ and all $i \in \mathbb{Z}_+$. Therefore, it trivially follows that there exists a compact set $Q $ such that $\mu_{n+1}^i(Q) > 1- \epsilon$ and $\nu^i_{n}(Q \times U) > 1- \epsilon$. This implies that the set of measures that satisfy the constraints $\nu^i_{n}(A \times U) = \mu^i_{n}(A)$ for all $A \in \mathcal{B}(X)$ and all $i \in \mathbb{Z}_+$ is {\it tight} \cite{billingsley2013convergence}, and therefore is relatively compact, i.e, every sequence of measures $(\mu_{n+1}^i,\nu_{n}^i)$ contains a narrowly converging subsequence, also denoted by $(\mu_{n+1}^i,\nu_{n}^i)$, for each $n \in \{0,1,...,N-1\}$. Since the map $T:X \times U \rightarrow X$ is continuous, the map $T_{\#} : \mathcal{P}(X \times U) \rightarrow \mathcal{P}(X)$  is narrowly continuous. Hence, for each $n \in \lbrace 0,1,...,N-1\rbrace$, there exists a limit measure $\nu_n \in \mathcal{P}(X \times U)$ such that $T_{\#} \nu^i_n$ narrowly converges to a unique limit $T_{\#} \nu_{n}$ as $i \rightarrow \infty$.  Moreover, it also follows that the subsequence of marginal measures $\nu^i_{n}(\cdot \times U) = \mu^j_{n}$ narrowly converges  to the unique limit  $\mu_{n}$ for each $n \in\lbrace 0,1,..., N-1 \rbrace$.
\end{proof}

From the above proposition, we obtain one of the main results of this paper.

\begin{theorem}
	\label{ctrthe}
	Let $\mu_0, \mu^f \in \mathcal{P}(X)$ be Borel probability measures with compact supports, such that ${\rm supp} ~ \mu^f \subseteq R_N^{\mathbf{x}}$  for each $\mathbf{x} \in {\rm supp}~ \mu_0$. Then there exists a sequence of stochastic feedback laws $(K_n)_{m=1}^{N-1} \in \mathcal{Y}(X,U)$ such that the system of equations \eqref{eq:ctrbcond} is satisfied, 
    and hence the measure $\mu^f$ can be reached from the measure $\mu_{0}$.
\end{theorem}

\begin{proof}
	Note that $X$ and $U$ are separable. Hence, the product $\sigma$-algebra on $X \times U$ is equal to $\mathcal{B}(X \times U)$. Then, given a measure $\nu \in \mathcal{P}(X \times Y)$, from the {\it disintegration theorem} \cite{florescu2012young}[Theorem 3.2] there exists a measure $\mu \in \mathcal{P}(X)$ and stochastic feedback law $K \in \mathcal{Y}(X,U)$ such that 
\begin{equation}
\int_{A \times B}d\nu(\mathbf{x},\mathbf{u}) =\int_{A} \int_B K(\mathbf{x},d\mathbf{u}) d\mu(\mathbf{x})
\end{equation}
for all $A \in \mathcal{B}(X)$ and all $B \in \mathcal{B}(U)$. Then the result follows from Proposition \ref{prodCtrb}. In particular, using the measures $(\nu_m)_{m=0}^{N-1} \in \mathcal{P}(X \times U)$, by disintegration, the stochastic feedback laws $(K_m)_{m=0}^{N-1} \in \mathcal{Y}(X,U)$ can be constructed such that the system of equations \eqref{eq:ctrbcond} holds true.
\end{proof}


\begin{remark} (\textbf{Conservatism of controllability result})
Theorem \ref{ctrthe} gives a  sufficient, but not necessary, condition on 
system \eqref{eq:origsys} for Problem \ref{reachprob} to admit a solution: namely, that each point in the support of the target measure be reachable from each point in the support of the initial measure. 
The controllability result in Theorem \ref{ctrthe} is conservative because we do not, in general, require this condition. To see this explicitly, consider the trivial example where $X = \mathbb{R}$, $U = \lbrace 0  \rbrace$, and $T(x,u) = x+u$. Suppose we define the initial and target measures as $\mu_0 = \mu^{f} = \frac{1}{2}\delta_{x_1} + \frac{1}{2} \delta_{x_2}  $ for some $x_1  \neq x_2$ in $\mathbb{R}$. Then it is straightforward to see that the target measure is reachable from the initial measure. However, the system is nowhere controllable in $\mathbb{R}$. More specifically, the points $x_1$ and $x_2$ are not reachable from each other.
\end{remark}

\section{OPTIMAL CONTROL}

\begin{figure*}
	\centering
	\begin{subfigure}[t]{0.32\textwidth}
		\centering
		\includegraphics[trim = 0cm 2cm 0cm 0cm,  width= \textwidth]{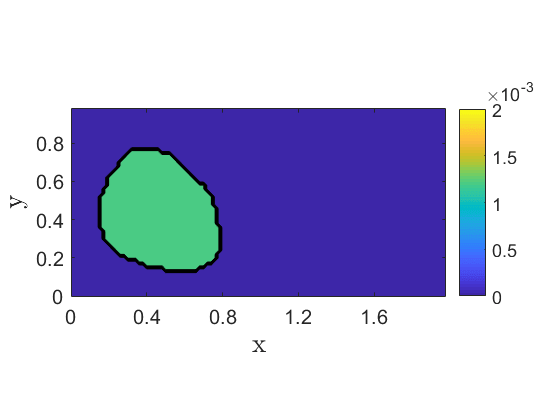}
		\caption{$n = 0$ (Initial measure)}
		\label{fig:dguni2}
	\end{subfigure}%
	~
	\begin{subfigure}[t]{0.32\textwidth}
		\centering
		\includegraphics[trim = 0cm 2cm 0cm 0cm, width= \textwidth]{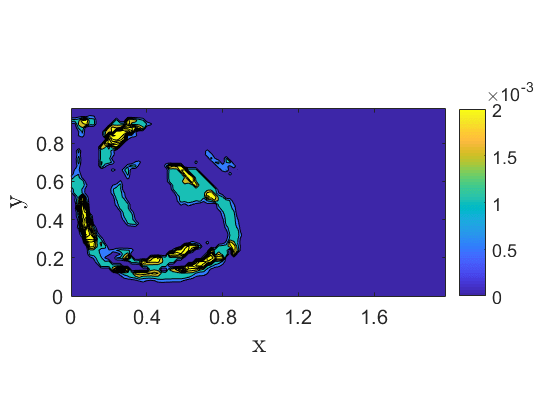}
		\caption{$n = 2$ }
		\label{fig:dguni2-min}
	\end{subfigure}
    ~
	\begin{subfigure}[t]{0.32\textwidth}
		\centering
		\includegraphics[trim = 0cm 2cm 0cm 0cm, width= \textwidth]{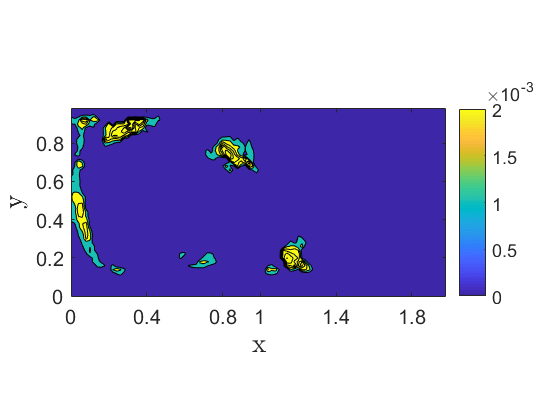}
		\caption{$n = 4$ }
		\label{fig:dguni2-min}
	\end{subfigure}
	\centering
	\begin{subfigure}[t]{0.32\textwidth}
		\centering
		\includegraphics[trim = 0cm 2cm 0cm 0cm, width= \textwidth]{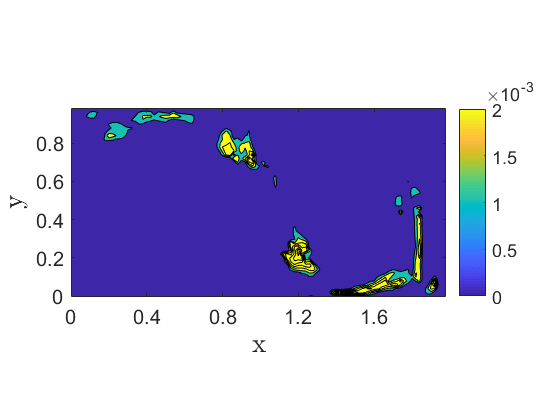}
		\caption{$n = 6$}
		\label{fig:dguni2}
	\end{subfigure}
	\begin{subfigure}[t]{0.32\textwidth}
		\centering
		\includegraphics[trim = 0cm 2cm 0cm 0cm, width= \textwidth]{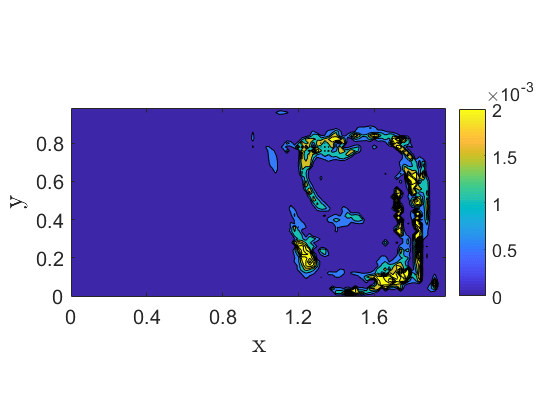}
		\caption{$n = 8$ }
		\label{fig:dguni2}
	\end{subfigure}
	~
	\begin{subfigure}[t]{0.32\textwidth}
		\centering
		\includegraphics[trim = 0cm 2cm 0cm 0cm, width= \textwidth]{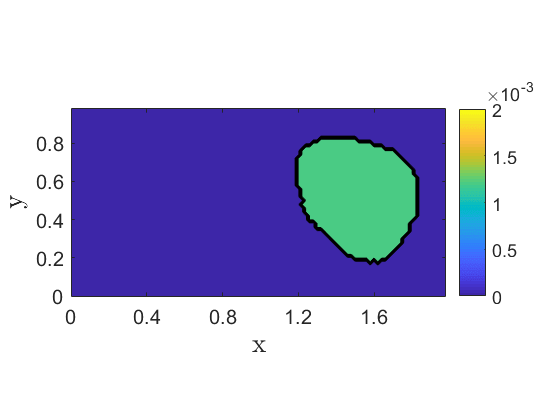}
		\caption{$n = 10$ (Final measure)}
		\label{fig:dguni3}
	\end{subfigure}
	\caption{Solution of the optimal transport problem at several times $n$ for unicycles in a double-gyre flow model}
	\label{fig:MainGraph1}
\end{figure*}

This section addresses Problem \ref{OCP}. As in the proof of the controllability result in Theorem \ref{ctrthe}, we will apply the disintegration theorem \cite{florescu2012young}[Theorem 3.2] to the correspondence between elements of $\mathcal{Y}(X,U)$ and elements of $\mathcal{P}(X \times U)$ with a given marginal. Hence, the optimization problem \eqref{eq:op1main1}-\eqref{eq:op1main2} 
can be convexified 
by replacing stochastic feedback laws $K_n \in \mathcal{Y}(X,U)$ with elements $\nu_n  \in \mathcal{P}(X \times U)$ and by enforcing appropriate constraints on the marginals of the measures $\nu_n$. These modifications allow us to frame the optimization problem in Problem \ref{OCP} as an equivalent infinite-dimensional linear programming problem: 
 \begin{eqnarray}
\label{eq:Optim1a}
\min_{\substack{\mu_{m+1} \in \mathcal{P}(X),\\ ~\nu_m \in \mathcal{P}(X \times U)}}~~ \sum_{m=0}^{N-1} \int_{X \times U} c(\mathbf{x},\mathbf{u}) d\nu_{m}(\mathbf{x},\mathbf{u})
\end{eqnarray}
subject to the constraints
\begin{align}
\label{eq:Optim1b}
\mu_{n+1} &= T_{\#}\mu_n  , ~~ n =0,1,...,N-1; ~~~~\mu_{N} = \mu^f,  \nonumber \\
\pi_{\#}\nu_{n} &= \mu_{n},  
\end{align}
where $\pi:X\times U \rightarrow X$ is the projection map defined by  $\pi(\mathbf{x},\mathbf{u})= \mathbf{x}$ for all $\mathbf{x} \in X$ and all $\mathbf{u} \in U$. Here, the constraints $\pi_{\#}\nu_{n} = \mu_{n}$ ensure that, for each $n \in \lbrace 0,1,...,N-1 \rbrace$, $\nu_{n}(A \times U) = (\pi_{\#}\nu_{n})(A) = \mu_{n}(A)$ for all $A \in \mathcal{B}(X)$. Hence, we have the following result. 

\begin{theorem}
	Let $\mu_0, \mu^f \in \mathcal{P}(X)$ be Borel probability measures with compact supports, such that ${\rm supp} ~ \mu^f \subseteq R_N^{\mathbf{x}}$  for each $\mathbf{x} \in {\rm supp}~ \mu_0$. Then the optimization problem \eqref{eq:Optim1a}-\eqref{eq:Optim1b} has a solution $ (\mu_{n+1},\nu_{n}), ~ n=0,1,...,N-1$.
    \label{thOptim}
\end{theorem}
\begin{proof}
The proof follows the standard compactness-based arguments in optimization.
From Theorem \ref{ctrthe}, we know that the set of measures satisfying  constraints \eqref{eq:Optim1b} is non-empty. 
Moreover, the map $c:X \times U \rightarrow \mathbb{R}$ is continuous. 
Since $T$ is continuous, measures with compact support are pushed forward to measures with compact support. This implies that for any choice of measure $\nu_{n}$, ${\rm supp} ~ \mu_{n+1}$ is contained in a compact set since ${\rm supp} ~ \mu_0$ is contained in a compact set. 
Therefore, $\sum_{m=0}^{N-1} \int_{X \times U} c(\mathbf{x},\mathbf{u}) d\nu_{m}(\mathbf{x},\mathbf{u})$ is bounded from below on the set of admissible measures. Hence, there exists a minimizing sequence of measures $(\mu_{n+1}^i,\nu^i_{n})_{i=1}^\infty $, with $(\mu_{n+1}^i,\nu^i_{n}) \in \mathcal{P}(X) \times \mathcal{P}(X \times U)$ for each $n \in \lbrace 0,1,...,N-1 \rbrace$, that satisfies the constraints \eqref{eq:Optim1b}. By {\it minimizing}, we mean that the sequence of measures $(\mu_{n+1}^i,\nu^i_{n})_{i=1}^\infty$ satisfies $ \lim_ {i \rightarrow \infty} \sum_{m=0}^{N-1} \int_{X \times U} c(\mathbf{x},\mathbf{u}) d\nu^i_m(\mathbf{x},\mathbf{u}) = \inf_{\mu_{m+1} \in \mathcal{P}(X), ~\nu_m \in \mathcal{P}(X \times U)}~~ \sum_{m=0}^{N-1} \int_{X \times U} c(\mathbf{x},\mathbf{u}) d\nu_m(\mathbf{x},\mathbf{u})$, with  the infimum taken over the constraint set \eqref{eq:Optim1b}. 
	We now confirm that there exist measures $(\mu_{n+1}^*,\nu^*_{n})$ that achieve this infimum. We recall that the support of the measures $(\mu_{n+1},\nu_{n})$ is compact for all $n \in \lbrace 0,1,...,N-1 \rbrace$ and that the set of measures that satisfy the constraints \eqref{eq:Optim1b} is relatively compact, i.e, every sequence of measures $(\mu_{n+1}^i,\nu_{n}^i)$ contains a narrowly converging subsequence $(\mu_{n+1}^i,\nu_{n}^i)$. The map $\gamma \mapsto \int_{X \times U} c(\mathbf{x},\mathbf{u}) d\gamma(\mathbf{x},\mathbf{u})$, a map from $\mathcal{P}(X \times U)$ to $\mathbb{R}$, is narrowly continuous. Hence, there exist limit measures $(\mu^*_{n+1},\nu^*_{n})$ such that $\sum_{m=0}^N \int_{X \times U} c(\mathbf{x},\mathbf{u}) d\nu^*_{m}(\mathbf{x},\mathbf{u}) = \inf_{\mu_{m+1} \in \mathcal{P}(X), ~\nu_m \in \mathcal{P}(X \times U)}~~ \sum_{m=0}^{N-1} \int_{X \times U} c(\mathbf{x},\mathbf{u}) d\nu_{m}(\mathbf{x},\mathbf{u})$, subject to the constraints \ref{eq:Optim1b}. This concludes the proof.
\end{proof}

By disintegration of the measures $\nu_m$ in Theorem \ref{thOptim}, it is straightforward to conclude the following result.

\begin{theorem}
	Let $\mu_0, \mu^f \in \mathcal{P}(X)$ be Borel probability measures with compact supports, such that ${\rm supp} ~ \mu^f \subseteq R_N^{\mathbf{x}}$  for each $ \mathbf{x} \in {\rm supp}~ \mu_0$. Then the optimization problem in Problem \eqref{OCP} has a solution $ (\mu_{n+1},K_{n}), ~n=0,...,N-1$.
\end{theorem}


\section{NUMERICAL OPTIMIZATION}

In this section, we briefly describe a numerical approach to solving the optimization problem in Problem \ref{OCP}. In both the examples that we consider in Section \ref{sec:nuex}, the state space $X$ is taken to be a compact subset of $\mathbb{R}^2$. This subset $X$ is partitioned into $n_x \in \mathbb{Z}_{+}$
sets, $\tilde{X} = \lbrace \Omega_1,..., \Omega_{n_x} \rbrace$, 
whose union is $X$ and whose intersections have zero Lesbesgue measure.
The set of control inputs $U$ is 
approximated as a set of $n_u \in \mathbb{Z}_{+}$ discrete elements,
$\tilde{U} = \lbrace \gamma_1,...,\gamma_{n_u} \rbrace$, 
where $\gamma_i \in U$ for each $i$. 
We then use the {\it Ulam-Galerkin method} \cite{bollt2013applied} to construct an approximating controlled Markov chain on a finite state space $V = \lbrace 1,...,n_x \rbrace $. 
In the uncontrolled setting, this method is a classical technique used to construct approximations of pushforward maps induced by dynamical systems, also known as Perron-Frobenius operators.

We define the {\it controlled} transition probabilities for the Markov chain on $V$ as follows:
\begin{equation}
\tilde{p}_{ij}^k = \frac{\tilde{m}(T_k^{-1}(\Omega_j) \cap \Omega_i)}{\tilde{m}(\Omega_i)}, \nonumber
\end{equation}
where $\tilde{m}$ is the Lebesgue measure and $T_k  =T(\cdot,\gamma_k)$. The quantity $\tilde{p}_{ij}^k$ is the probability of the system state entering the set $\Omega_j$ in the next time step, given that this state is uniformly randomly distributed over the set $\Omega_{i}$ (identified with $i \in V$) and the control input is chosen to be $\gamma_k$. We also define an equivalent of the stochastic feedback law $K_n$ in the discretized case that we consider. Toward this end, we denote by 
$\lambda_n^{k,i}$ the probability of choosing the 
control input $\gamma_k$, given that the system state is in $\Omega_i$
at time $n$. We define the variables $\tilde{\nu}_{n}^{k,i} = \tilde{\mu}_n^i \lambda_n^{k,i}$, where $\tilde{\mu}_n^i$ is the probability of the state being in $\Omega_i$ at time step time $n$. Additionally, let $\tilde{c}_{i,k} = \int_{\Omega_i}c(\mathbf{x},\gamma_k)d\mathbf{x}$ be the average cost of the state being in $\Omega_i$ and the control input given by $\gamma_k$. 
\begin{figure*}
	\centering
	\begin{subfigure}[t]{0.235\textwidth}
		\centering
		\includegraphics[width= \textwidth]{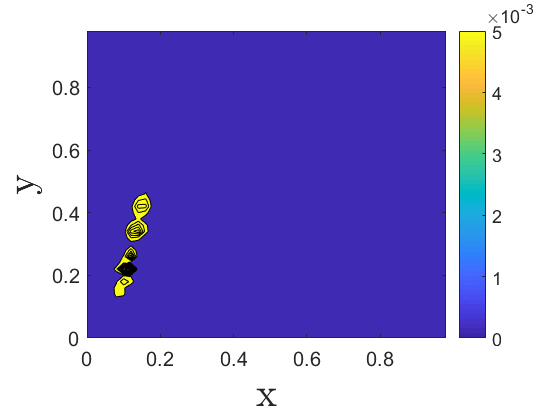}
		\caption{$n = 4$}
		\label{fig:dint4}
	\end{subfigure}%
	~
    \begin{subfigure}[t]{0.235\textwidth}
		\centering
		\includegraphics[width= \textwidth]{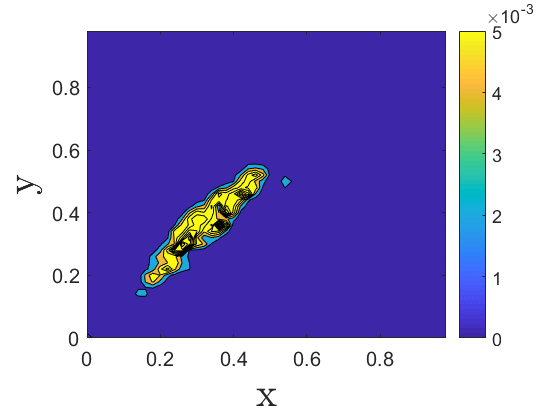}
		\caption{$n = 8$}
		\label{fig:dint8}
	\end{subfigure}%
	~
	\begin{subfigure}[t]{0.235\textwidth}
		\centering
		\includegraphics[width= \textwidth]{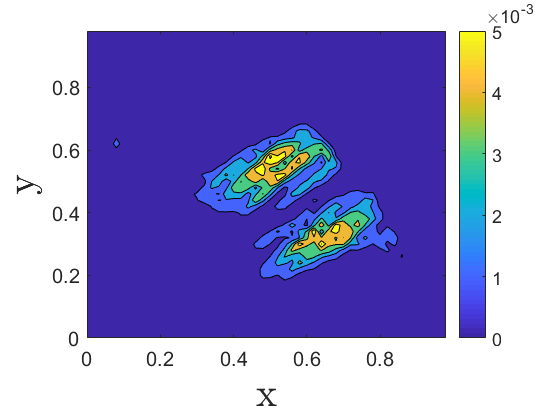}
		\caption{$n = 12$ }
		\label{fig:dint12}
	\end{subfigure}
    ~
    \begin{subfigure}[t]{0.235\textwidth}
		\centering
		\includegraphics[width= \textwidth]{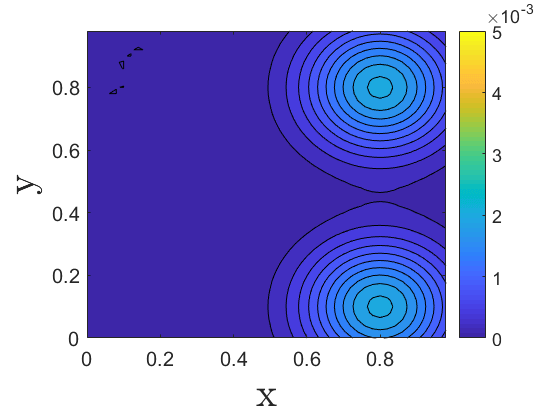}
		\caption{$n = 15$ }
		\label{fig:dint15}
	\end{subfigure} 
\caption{Solution of the optimal transport problem at several times $n$ for a  double-integrator system}
	\label{fig:MGDint}
\end{figure*}


Given these parameters and specified initial and target measures $\tilde{\mu}_{0}, \tilde{\mu}^f \in \mathcal{P}(\tilde{X})$, we can define the finite-dimensional equivalent of the linear programming problem \eqref{eq:Optim1a}-\eqref{eq:Optim1b} as follows: 
 \begin{eqnarray}
\min_{\tilde{\mu}_{m+1}^i, \tilde{\nu}_{m}^{k,i} \in \mathbb{R}_{\geq 0} }~~ \sum_{m=0}^{N-1} \sum_{i=1}^{n_x} \sum_{k=1}^{n_u} \tilde{c}_{i,k} \tilde{\nu}_{m}^{k,i} 
\end{eqnarray}
subject to the constraints
\begin{eqnarray}
&\tilde{\mu}^j_{n+1} = \sum_{k = 1}^{n_u}\sum_{i = 1}^{n_x} \tilde{p}^{k}_{ij}\tilde{\nu}_{n}^{k,i}, \nonumber \\
 &\tilde{\mu}_{N}^j = (\tilde{\mu}^f)^j,  \nonumber \\
&\sum_{i=1}^{n_x} \tilde{\mu}^i_{n+1} = 1, ~ \sum_{k=1}^{n_u} \tilde{\nu}_n^{k,j} = \tilde{\mu}^j_n, 
\end{eqnarray}
for $n \in \lbrace  0,..., N-1 \rbrace$ and  $j \in \lbrace  1,..., n_x \rbrace$.

After solving this linear programming problem, we can extract the control laws $\lambda_n^{k,i}$ by setting $\lambda_n^{k,i} = \frac{\tilde{\nu}_{n}^{k,i}}{\tilde{ \mu}_n^i}$ if  $\tilde{\mu}_n^i \neq 0 $ and $\lambda_n^{k,i} = 0$ otherwise. The resulting Markov chain evolves according to the equation $ \tilde{\mu}^j_{n+1} = \sum_{k = 1}^{n_u}\sum_{i = 1}^{n_x} \tilde{p}^{k}_{ij}\lambda_n^{k,i} \tilde{\mu}^i_n$. 

\section{SIMULATION EXAMPLES}\label{sec:nuex}

In this section, we apply the numerical optimization procedure described in the previous section to two examples. Neither example can be solved by classical optimal transport methods, due to the nonlinearity of the control system (Example 1) or the bounds on the control set (Examples 1 and 2). In both examples, we define the cost function as $c(\mathbf{x},\mathbf{u}) = \|\mathbf{x}\|^2 +\|\mathbf{u}\|^2 $, where $\|\cdot\|$ represents the $2$-norm.

\subsection{Example 1: Unicycles in a Time-Periodic Double Gyre}
We consider the system
\begin{equation}
\mathbf{x}_{n+1}  = F(\mathbf{x}_n) + G(\mathbf{u}),
\end{equation}
where $\mathbf{x}_n = [x_n~y_n]^T \in X$, $\mathbf{u}= [u^1 ~ u^2]^T \in U$, and 
$G(\mathbf{u}) = [u^1 \cos(u^2) ~~ u^1 \sin(u^2)]^T$. The phase space is $X=[0,2]\times[0,1]$, and the set of control inputs is $U = [-1,1] \times [0,2  \pi]$. The final time is set to $N = 10$.
To define the map $F: X \rightarrow X$, we consider the double-gyre system \cite{elamvazhuthi2016optimal}:
\begin{align}
\dot{x}&=-\pi A\sin(\pi f(x,t))\cos(\pi y), \label{eq:dg1} \\
\dot{y}&=\pi A\cos(\pi f(x,t))\sin(\pi y)\dfrac{df(x,t)}{dx}, \label{eq:dg2}
\end{align}
where $f(x,t)=\beta\sin(\omega t)x^2+(1-2\beta\sin(\omega t))x$ is the time-periodic forcing in the system. The map $F$ is defined by setting $F(\mathbf{x})$ equal to the solution of
equations \eqref{eq:dg1}-\eqref{eq:dg2}, integrated over the time period $\tau$. In this example, we define $A=0.25$, $\beta=0.25$, and $\omega=2\pi$, which results in $\tau=1$. The set $X$ is not invariant for all choices of control inputs in $U$. Hence, since this set must be approximatable by a finite set, we define $F(\mathbf{x}) + G(\mathbf{u}) \triangleq \mathbf{x}$ if  $F(\mathbf{x}) + G(\mathbf{u}) \notin X$ 
for some $(\mathbf{x},\mathbf{u}) \in X \times U$.
The initial and target measures are chosen to be uniform over certain \emph{almost-invariant sets} \cite{bollt2013applied} in the left and right halves of the domain, respectively. The optimal transport shown in Fig. \ref{fig:MainGraph1} exploits \emph{lobe dynamics}, i.e., the control inputs push the initial measure onto regions bounded by stable and unstable manifolds. As a result, the measure is transported into the right half of the domain under the action of $F$.
\vspace{-1mm}

\subsection{Example 2: Double-Integrator System}
In this example, we consider the following system: 
\vspace{-1mm}
\begin{equation}\label{eq:dintsys}
x_{n+1} = x_n + 0.15 y_n, ~~~
y_{n+1} = y_n + u, 
\end{equation}
with $[x_n ~y_n]^T \in X = [0,1]^2$ and $ u \in U = [-0.25,0.25]$.
The final time is set to $N = 15$. For unbounded control inputs, this control system can be verified to be globally controllable using the Kalman rank condition. For compact control sets, controllability is harder to verify without numerical computation. The initial measure is taken to be the Dirac measure concentrated at $[0~0]^T \in X$. 
The target measure is a linear combination of Gaussian distributions that are centered at $[0.8~0.1]^T$ and $[0.8~0.8]^T$, as shown in Fig. \ref{fig:dint15}. Measures at three intermediate times are shown in Fig. \ref{fig:dint4}-\ref{fig:dint12}. The control map adds a ``drift'' term $0.15 y_n$ to $x_{n+1}$ in equation \eqref{eq:dintsys}, which makes the system controllable despite the fact that it is underactuated. Figure \ref{fig:MGDint} confirms that this drift drives the initial measure exactly to the target measure at $N=15$.

\section{CONCLUSIONS}
In this paper, we have presented a relaxed version of the optimal transport problem for discrete-time nonlinear systems. We showed that under mild assumptions on the controllability of the original system, the extended system on the space of measures is controllable. This enabled us to prove the existence of solutions of an optimal transport problem for discrete-time nonlinear systems. 
One direction for future work is to explore conditions under which deterministic feedback maps exist for the optimal transport problem. Another interesting question is whether one can provide guarantees on the performance of the controllers obtained by solving the numerical optimization problem when these controllers are implemented on the original nonlinear system.






\bibliographystyle{plain}
\bibliography{StTrans}

\end{document}